\documentclass[11pt]{article}
 
\usepackage[a4paper, total={6in, 8in}]{geometry}
\usepackage{amsmath}
\usepackage{mathtools}
\usepackage{amsthm}
\usepackage{amsfonts}
\usepackage{amssymb}
\usepackage{xcolor}
\usepackage{comment}
\usepackage[ruled,vlined]{algorithm2e}
\usepackage{cleveref}
\usepackage{amsmath}
\usepackage{tikz}
\usepackage{tabularx}
\usepackage[ruled,vlined]{algorithm2e}
\usepackage[noend]{algpseudocode}
\LinesNumbered
\usepackage[multiple]{footmisc}
\newtheorem{theorem}{Theorem}[section]
\newtheorem{question}[theorem]{Question}
\newtheorem{fact}[theorem]{Fact}
\newtheorem{corollary}[theorem]{Corollary}
\newtheorem{lemma}[theorem]{Lemma}
\newtheorem{claim}[theorem]{Claim}
\newtheorem{definition}[theorem]{Definition}
\newtheorem{observation}[theorem]{Observation}
\newcommand{\bin}{\{0,1\}}
\newcommand{\D}{\mathsf{D}}
\newcommand{\R}{\mathsf{R}}

\newcommand{\RS}{\mathsf{RS}}
\newcommand{\product}{\mathsf{prod}}
\newcommand{\E}{\mathbf{E}}
\newcommand{\I}{\mathsf{Inf}}
\newcommand{\Var}{\mathsf{Var}}
\newcommand{\Qu}{\mathsf{Q}}

\newcommand{\Dprod}[1]{\mathsf{\overline{\D_{#1}^{\product}}}}
\newcommand{\Dprodmu}[2]{\mathsf{\overline{\D_{#1}^{#2}}}}
\newcommand{\A}{\mathcal{A}}

\newcommand{\Rprodnb}[1]{\R_{#1}^{\product}}
\newcommand{\Dprodnb}[1]{\D_{#1}^{\product}}
\newcommand{\Prb}{\mathcal{P}}
\allowdisplaybreaks

\title{Randomized query composition and product distributions}

\author{
Swagato Sanyal\footnote{IIT Kharagpur, \textsf{swagato@cse.iitkgp.ac.in}.}\footnote{S.S.~is supported by an ISIRD Grant from Sponsored Research and Industrial Consultancy, IIT Kharagpur.}}
\date{}

\begin{document}

\maketitle

\begin{abstract}
Let $\R_\epsilon$ denote randomized query complexity for error probability $\epsilon$, and $R:=\R_{1/3}$. In this work we investigate whether a perfect composition theorem $\R(f \circ g^n)=\Omega(\R(f)\cdot \R(g))$ holds for a relation $f \subseteq \bin^n \times \mathcal{S}$ and a total inner function $g:\bin^m \to \bin$.

Composition theorems of the form $\R(f \circ g^n)=\Omega(\R(f)\cdot \mathsf{M}(g))$ are known for various measures $\mathsf{M}$. Such measures include the sabotage complexity $\RS$ defined by Ben-David and Kothari (ICALP 2015), the max-conflict complexity defined by Gavinsky, Lee, Santha and Sanyal (ICALP 2019), and the linearized complexity measure defined by Ben-David, Blais, G{\"o}{\"o}s and Maystre (FOCS 2022). The above measures are asymptotically non-decreasing in the above order. However, for total Boolean functions no asymptotic separation is known between any two of them.

Let $\Dprodnb{}$ denote the maximum distributional query complexity with respect to any product (over variables) distribution . In this work we show that for any total Boolean function $g$,  the sabotage complexity $\RS(g)=\widetilde{\Omega}(\Dprodnb{}(g))$. This gives the composition theorem $\R(f \circ g^n)=\widetilde{\Omega}(\R(f)\cdot \Dprodnb{}(g))$. In light of the minimax theorem which states that $\R(g)$ is the maximum distributional complexity of $g$ over any distribution, our result makes progress towards answering the composition question.

We prove our result by means of a complexity measure $\Rprodnb{\epsilon}$ that we define for total Boolean functions. Informally, $\Rprodnb{\epsilon}(g)$ is the minimum complexity of any randomized decision tree with unlabelled leaves with the property that, for every product distribution $\mu$ over the inputs, the average bias of its leaves is at least $((1-\epsilon)-\epsilon)/2=1/2-\epsilon$. It follows by standard arguments that $\Rprodnb{1/3}(g) = \Omega(\Dprodnb{}(g))$. We show that $\Rprodnb{1/3}$ is equivalent to the sabotage complexity up to a logarithmic factor.

Ben-David and Kothari asked whether $\RS(g)=\Theta(\R(g))$ for total functions $g$. We generalize their question and ask if for any error $\epsilon$, $\Rprodnb{\epsilon}(g)=\widetilde{\Theta}(\R_\epsilon(g))$. We observe that the work by Ben-David, Blais, G{\"o}{\"o}s and Maystre (FOCS 2022) implies that for a perfect composition theorem $\R_{1/3}(f \circ g^n)=\Omega(\R_{1/3}(f)\cdot\R_{1/3}(g))$ to hold for any relation $f$ and total function $g$, a necessary condition is that $\R_{1/3}(g)=O(\frac{1}{\epsilon}\cdot \R_{\frac{1}{2}-\epsilon}(g))$ holds for any total function $g$. We show that $\Rprodnb{\epsilon}(g)$ admits a similar error-reduction $\Rprodnb{1/3}(g)=\widetilde{O}(\frac{1}{\epsilon}\cdot\Rprodnb{\frac{1}{2}-\epsilon}(g))$. Note that from the definition of $\Rprodnb{\epsilon}$ it is not immediately clear that $\Rprodnb{\epsilon}$ admits any error-reduction at all.

We ask if our bound $\RS(g) = \widetilde{\Omega}(\Dprodnb{}(g))$ is tight. We answer this question in the negative, by showing that for the NAND tree function, sabotage complexity is polynomially larger than $\Dprodnb{}$. Our proof yields an alternative and different derivation of the tight lower bound on the bounded error randomized query complexity of the NAND tree function (originally proved by Santha in 1985), which may be of independent interest. Our result shows that sometimes, $\Rprodnb{1/3}$ and sabotage complexity may be useful in producing an asymptotically larger lower bound on $\R(f \circ g^n)$ than $\widetilde{\Omega}(\R(f)\cdot \Dprodnb{}(g))$. In addition, this gives an explicit polynomial separation between $\R$ and $\Dprodnb{}$ which, to our knowledge, was not known prior to our work.
\end{abstract}

\section{Introduction}
\label{sec:intro}
A decision tree or a query algorithm for a relation $f \subseteq \bin^n \times S$ can query various bits of an input bit string $x=(x_1, \ldots, x_n)$ in an adaptive fashion, with the goal of outputting an $s \in S$ such that $(x, s) \in f$. A randomized decision tree is assumed to have access to some source of randomness, and may choose a variable to query based on responses to previous queries, and the randomness. The complexity of a decision tree is the number of variables that it queries in the worst case. A decision tree that uses no randomness and for every $x$ outputs an $s$ such that $(x, s)\in f$ is called a deterministic decision tree computing $f$. The randomized query complexity of $f$ for error $\epsilon$, denoted by $\R_\epsilon(f)$, is the least complexity of any randomized decision tree that, for every input $x$, outputs $s$ such that $(x, s) \in f$ with probability (over its own randomness) at least $1-\epsilon$. Similarly, the deterministic query complexity of $f$, denoted by $\D(f)$, is the least complexity of any deterministic decision tree computing $f$. For a probability distribution $\mu$ over the domain of $f$, the distributional query complexity of $f$ with respect to $\mu$ and for error $\epsilon$, denoted by $\D_\epsilon^\mu(f)$, is the least complexity of any deterministic decision tree that, for a random input $x$ sampled from $\mu$, fails to output an s such that $(x, f) \in f$ with probability at most $\epsilon$. Define $\R(f):=\R_{1/3}(f)$ and $\D^\mu(g):=\D^\mu_{1/3}(g)$.

For a total Boolean function $g:\bin^m \to \bin$, the composition $f \circ g^n$ is the relation comprising all pairs $((x_1, \ldots, x_n), s) \in (\bin^m)^n$ such that $((g(x_1), \ldots, g(x_n)), s) \in f$.

It is easy to see that $\D(f \circ g^n) \leq \D(f)\cdot \D(g)$; a decision tree for $f \circ g^n$ may be constructed simply by simulating an optimal tree of $f$, and serving each query that it makes by solving the corresponding copy of $g$ using an optimal tree of $g$. For randomized query algorithms, a similar idea works out, albeit with some additional work to handle errors, to show that $\R(f \circ g^n)=O(\R(f)\cdot\R(g)\cdot \log \R(f))$.

Composition questions ask whether the aforementioned upper bounds on the complexity of $f \circ g^n$ are asymptotically optimal. These are fundamental questions about the structure of optimal algorithms for $f \circ g^n$, and have received considerable attention in research.

It is known from the works of Montenaro \cite{M14} and Tal \cite{T13} that $\D(f \circ g^n) = \D(f)\cdot \D(g)$. Thus the composition question for deterministic query complexity has been completely answered. On the contrary, in spite of extensive research, a complete answer to the composition question for randomized query complexity is still lacking.
\subsection{Past works on randomized query composition}
Past works dealt with a more general composition question for randomized query complexity, where the inner function $g$ is allowed to be partial. The definition of $f \circ g^n$ and the aforementioned upper bounds on $\D(f \circ g^n)$ and $\R(f \circ g^n)$ can be accordingly generalized; we are omitting the details in this paper. In 2015, Ben-David and Kothari \cite{BK15} defined the sabotage complexity measure $\RS(g)$ of a partial Boolean function $g$. They showed that $\R(f \circ g^n) = \Omega(\R(f)\cdot \RS(g))$. They further showed that for total $g$, $\RS(g) =\widetilde{\Omega}(\sqrt{\R(g)})$, implying $\R(f \circ g^n) = \widetilde{\Omega}(\R(f)\cdot \sqrt{\R(g)})$. In 2019, Gavinsky, Lee, Santha and Sanyal \cite{GLSS19} introduced the max-conflict complexity $\overline{\chi}(g)$ and showed that $\R(f \circ g^n) = \Omega(\R(f)\cdot \overline{\chi}(g))$. They further showed that even for partial functions $g$, $\overline{\chi}(g)=\Omega(\sqrt{\R(g)})$, implying $\R(f \circ g^n) = \Omega(\R(f)\cdot \sqrt{\R(g)})$. Moreover, they showed that for all partial functions $g$, $\overline{\chi(g)}=\Omega(\RS(g))$. They also demonstrated unbounded separation between $\overline{\chi}(g)$ and $\RS(g)$ for a partial $g$. In 2022, Ben-David, Blais, G{\"o}{\"o}s and Maystre \cite{BBGM22} introduced the linearized complexity measure $\mathsf{L}(g)$. They showed that for any partial $g$, $\R(f \circ g^n) = \Omega(\R(f)\cdot \mathsf{L}(g))$, and that $\mathsf{L}$ is the largest measure $\mathsf{M}$ for which the statement $\R(f \circ g^n) = \Omega(\R(f)\cdot \mathsf{M}(g))$ holds. They also demonstrated polynomial separation between $\mathsf{L}(g)$ and $\overline{\chi}(g)$ for a partial $g$.

A different line of work has focused on proving bounds on $\R(f \circ g^n)$ of the form $\Omega(\mathsf{M}(f)\cdot \R(g))$ for some complexity measure $\mathsf{M}$ \cite{AGJKLMSS18, GJPW18, BDGHMT20}. In 2020 Ben-David and Blais \cite{BB20} defined the noisy query complexity $\mathsf{noisyR}$ and showed that $\mathsf{noisyR}$ is the largest measure $\mathsf{M}$ for which the statement $\R(f \circ g^n)=\Omega(\mathsf{M}(f)\cdot \R(g))$ holds. In 2023, Chakraborty et~al. \cite{CKMPSS23} showed that for the special case when $\R(f)=\Theta(n)$, a near-perfect randomized query composition theorem $\R(f \circ g^n)=\widetilde{\Omega}(\R(f)\cdot \R(g))$ holds.
\subsection{Our results}
This work investigates the possibility of a perfect randomized query composition theorem $\R(f \circ g^n)=\Omega(\R(f) \cdot \R(g))$ when $g$ is a total function. As discussed in the preceding section, past works have introduced measures $\RS, \overline{\chi}$ and $\mathsf{L}$ that are asymptotically non-decreasing in the above order. As discussed before, we also know that any two of them can be asymptotically separated. However, the Boolean functions that witness these separations are all partial, and to the best of our knowledge, no separation between these measures is known for total functions. Does one of these measures coincides with $\R$ for total functions?

Ben-David and Kothari asked in their paper whether $\RS(g)=\Theta(\R(g))$ for total $g$. Our first result is that for any total $g$, $\RS(g)$ is, up to a logarithmic factor, at least the maximum distributional query complexity of $g$ for any product (over variables) distribution. Let \textsf{PROD} be the set of all product distributions over the domain $\bin^m$ of $g$. Define $\Dprodnb{}(g):=\max_{\mu \in \mathsf{PROD}}\{\D^\mu(g)\}$. 
\begin{theorem}
\label{thm:th1}
For any total function $g:\bin^m \to \bin$, $$\RS(g)=\widetilde{\Omega}(\Dprodnb{}(g)).$$
\end{theorem}
Informally, the sabotage complexity captures the minimum number of randomized queries required to distinguish any pair of input strings on which the function values differ (see Section~\ref{par:sabotage} for a formal definition). Theorem~\ref{thm:th1} shows that this task is at least as hard as deciding the function on every possible product distribution (potentially with a different query algorithm for each distribution).

Together with the composition theorem of Ben-David and Kothari, Theorem~\ref{thm:th1} immediately yields the following corollary.
\begin{corollary}
\label{cor:toth1}
For any total function $g:\bin^m \to \bin$, $$\R(f \circ g^n)=\widetilde{\Omega}(\R(f)\cdot \Dprodnb{}(g)).$$
\end{corollary}
The minimax theorem (Fact~\ref{fact:minimax}) states that $\R(g)=\max_{\mu} D^\mu(g)$, where the maximum is over all probability distributions over the  domain of $g$. In this light Corollary~\ref{cor:toth1} makes progress towards answering the randomized composition question for total inner functions. An additional motivation for our first result is that product distributions comprise a natural class of distributions that has received significant attention in Boolean function complexity research \cite{JKKLSSV20, HJR16, K16, S02}.

We prove Theorem~\ref{thm:th1} by introducing a new complexity measure $\Rprodnb{\epsilon}$. Informally speaking, $\Rprodnb{\epsilon}(g)$ is the minimum complexity of any randomized decision tree with unlabelled leaves with the property that, for every product distribution $\mu$ over the inputs, the average bias of its leaves is at least $((1-\epsilon)-\epsilon)/2=1/2-\epsilon$. Define $\Rprodnb{}(g):=\Rprodnb{1/3}(g)$. See Section~\ref{par:rprod} for formal definitions. It follows by standard arguments that $\Rprodnb{}(g) = \Omega(\Dprodnb{}(g))$ (see Claim~\ref{clm:rprod>=dprod}). Our next next result shows that $\RS$ is characterized by $\Rprodnb{}$ up to a logarithmic factor.
\begin{theorem}
\label{thm:th2}
For all total functions $g:\bin^m \to \bin$, 
\begin{enumerate}
    \item $\RS(g)=O(\Rprodnb{}(g))$, and
    \item $\RS(g)=\Omega(\Rprodnb{}(g) / \log \Rprodnb{}(g))$.
\end{enumerate}
\end{theorem}

\noindent Theorem~\ref{thm:th1} follows immediately from Theorem~\ref{thm:th2}(2) and the fact that $\Rprodnb{}(g) = \Omega(\Dprodnb{}(g))$ (Claim~\ref{clm:rprod>=dprod}).

Since any non-trivial product distribution is supported on all of $\bin^m$, $\Rprodnb{}(g)$ and $\Dprodnb{}(g)$ are well-defined only for total functions $g$. The proof of Theorem~\ref{thm:th2} (that goes via Lemma~\ref{lem:main1} discussed later) makes important use of the totality of $g$. We hope that the measure $\Rprodnb{}$, the characterization of $\RS$ presented in Theorem~\ref{thm:th2}, and the insights acquired in our proof techniques, specially pertaining to ways of exploiting totality, will be useful in future research to resolve whether $\RS(g)=\Theta(\R(g))$ for total functions $g$.

In light of Theorem~\ref{thm:th2} the question whether $\R(g)=\Theta(\RS(g))$ for total functions $g$ translates to the question whether $\R(g)=\widetilde{\Theta}(\Rprodnb{}(g))$. We generalize this question for every error $\epsilon$.
\begin{question}
\label{qn:qn1}
Is it true that for every total function $g:\bin^m \to \bin$ and $\epsilon: \mathbb{N} \to (0, 1/2)$, $\R_{\epsilon(m)}(g)=\widetilde{\Theta}(\Rprodnb{\epsilon(m)}(g))?$
\end{question}
From the work of Ben-David, Blais, G{\"o}{\"o}s and Maystre \cite{BBGM22} it follows that for any error parameter $\epsilon$, the linearized complexity measure $\mathsf{L}(g)$ of $g$ is bounded above by $O\left(\frac{1}{\epsilon}\cdot\R_{\frac{1}{2}-\epsilon}(g)\right)$. As discussed before, they also show that $\mathsf{L}$ is the largest measure $\mathsf{M}$ for which the statement $\R(f \circ g)=\Omega(\R(f) \mathsf{M}(g))$ holds for all relations $f$ and partial functions $g$. We thus have that for a perfect composition theorem $\R(f \circ g)=\Omega(\R(f)\R(g))$ to hold for any relation $f$ and any total Boolean function $g$, a necessary condition is that $\R(g)=O(\frac{1}{\epsilon}\cdot \R_{\frac{1}{2}-\epsilon}(g))$ holds for any total Boolean function $g$. In light of Question~\ref{qn:qn1}, we may ask if $\Rprodnb{\epsilon}$ admits a similar error reduction. Our next result answers this question in the affirmative (up to a logarithmic factor).
\begin{theorem}
\label{thm:th3}
For every total function $g:\bin^m \in \bin$ and $\epsilon:\mathbb{N} \to (0, 1/2)$,
\[\Rprodnb{}(g)=\widetilde{O}\left(\frac{1}{\epsilon(m)}\cdot\Rprodnb{\frac{1}{2}-\epsilon(m)}(g)\right).\]
\end{theorem}
We remark here that from the definition of $\Rprodnb{\epsilon}$ it is not immediately clear that it admits any error-reduction at all.

To prove Theorems~\ref{thm:th2} and~\ref{thm:th3}, we define a version of sabotage complexity with errors, that we denote by $\RS_\epsilon$. Informally, $\RS_\epsilon(g)$ is the minimum number of randomized queries required to distinguish every pair of inputs with different function values with probability at least $1-\epsilon$ (see Section~
\ref{par:sabotage} for a formal definition). Let $s(g)$ denote the sensitivity of $g$ (see Section~\ref{sec:influence} for a formal definition). The following lemma constitutes the technical core of our proofs of Theorems~\ref{thm:th2} and~\ref{thm:th3}.
\begin{lemma}
\label{lem:main1}
For all total Boolean functions $g:\bin^n \to \bin$, and $\epsilon: \mathbb{N} \to (0, 1/2)$,
\begin{enumerate}
    \item $\R^{\product}(g)=O\left(\frac{1}{\epsilon(n)}\cdot \RS_{1-\epsilon(n)}(g) \log s(g)\right)$, and
    \item $\RS_{1-2\epsilon(n)}(g) \leq \R^{\product}_{\frac{1}{2}-\epsilon(n)}(g)$.
\end{enumerate}
\end{lemma}

\noindent Is the bound in Theorem~\ref{thm:th1} tight? Our next result gives a negative answer to this question. We show that for the NAND tree function (defined shortly), $\RS$ and $\Dprodnb{}$ are polynomially separated. Consider a complete binary tree of depth $d$. Each leaf is labelled by a distinct Boolean variable. Each internal node is a binary NAND gate. For each input, the evaluation of this Boolean formula is the output of the NAND tree function, that we denote by $g_d$.
\begin{theorem}
\label{thm:th4}
$\Dprodnb{}(g_d)=O(\RS(g_d)^{1-\Omega(1)})$.
\end{theorem}
Saks and Wigderson \cite{SW86} showed that the zero-error randomized query complexity of $g_d$ is $\Theta(\alpha^d)$ for $\alpha=\frac{1+\sqrt{33}}{4}$.  Later Santha \cite{S91} showed that $\R(g_d)=\Theta(\alpha^d)$. We prove Theorem~\ref{thm:th4}  in two parts. First, we show an upper bound of $O((\alpha-\Omega(1))^d)$ on $\Dprodnb{}(g_d)$.
\begin{lemma}
\label{lem:main21}
There exists a constant $\delta >0$ such that $\Dprodnb{}(g_d)=O((\alpha-\delta)^d)$.
\end{lemma}
Works by Pearls \cite{P82} and Tarsi \cite{T13} showed that there exists a constant $\eta>0$ such that for all distributions $\mu$ where each variable is set to $1$ independently with some probability $p$, $\D^\mu(g_d)=O((\alpha-\eta)^d)$. In Lemma~\ref{lem:main21} we bound $\D^\mu(g_d)$ for any product distribution $\mu$. Our bound is quantitatively weaker than those by Pearls \cite{P82} and Tarsi \cite{T13}, and we do not comment on its tightness.

Lemma~\ref{lem:main21} also gives an explicit polynomial separation between $\R$ and $\Dprodnb{}$ which, to our knowledge, was not known prior to our work.

Next, we prove a tight lower bound on $\RS(g_d)$. As a by-product, our proof of the following lemma yields a different proof of the bound $\R(g_d)=\Omega(\alpha^d)$ from the one by Santha \cite{S91}, and may be of independent interest.
\begin{lemma}
\label{lem:main22}
$\RS(g_d)=\Omega(\alpha^d)$.
\end{lemma}
Lemma~\ref{lem:main22}, together with the upper bound by Saks and Wigderson, shows that $\RS(g_d)=\Theta(\R(g_d))$. From the composition theorem proven by Ben-David and Kothari, we thus have that for all relations $f$, $\R(f \circ g_d)=\Theta(\R(f)\cdot\R(g_d))$.

Lemmata~\ref{lem:main21} and~\ref{lem:main22} immediately imply Theorem~\ref{thm:th4}.
\subsection{Proof ideas}
In this section, we sketch the ideas and techniques that have gone into the proofs of our results. We begin with Lemma~\ref{lem:main1}, from which Theorems~\ref{thm:th2} and~\ref{thm:th3} follow. We then discuss Lemmata~\ref{lem:main21} and~\ref{lem:main22}, from which Theorem~\ref{thm:th4} follows.
\paragraph*{Lemma~\ref{lem:main1}}
We first discuss part 2, which is easier. Note that if a randomized algorithm $\mathcal{R}$ decides $g$ on each input with error probability $\frac{1}{2}-\epsilon$, then by a union bound two simultaneous runs of $\mathcal{R}$ on $x \in g^{-1}(0), y \in g^{-1}(1)$ decide both $g(x)$ and $g(y)$ with error probability $1-2\epsilon$. This implies that $\mathcal{R}$ distinguishes $x$ and $y$ with error probability $1-2\epsilon$.

Now we turn to part 1. This step needs arguments involving sensitivity and influence of Boolean functions, that are defined and discussed in Section~\ref{sec:influence}. The first step is showing that distinguishing each pair of inputs with high confidence is equivalent to reading each sensitive bit of each input with the same confidence (Lemma~\ref{lem:rs_and_sens_ip}). Using this, by a sequence of arguments involving standard error-reduction, we infer that there is a randomized tree $\mathcal{R}$ of complexity $O\left(\frac{1}{\epsilon(n)}\cdot\RS_{1-\epsilon(n)}(g)\log s(g)\right)$ that, for every input, with probability $1-\frac{1}{s(g)}$, queries all its sensitive bits. This translates to the claim that the average influence of the restrictions of $g$ to the leaves of $\mathcal{R}$ is low. Poincar\'{e} inequality (Lemma~\ref{lem:poincare}) now lets us conclude that that the average bias of those restrictions is small, yielding the lemma.
\paragraph*{Lemma~\ref{lem:main21}}
Saks and Wigderson gave a zero-error recursive algorithm for $g_d$. Their algorithm recursively evaluates a randomly chosen child of the root. If that child evaluates to $0$, the algorithm outputs $1$ and terminates. Else, the algorithm recursively evaluates the other child and outputs the complement.

If the output of the function is $0$, then the algorithm will be forced to evaluate both children. However, if the output is $1$, then the algorithm avoids evaluating one of the children with probability $1/2$.

We observe that if the inputs are sampled from a product distributions, then firstly, the output will not always be $0$; so we will always have scope to avoid evaluating one child. Secondly, we will also have both children evaluating to $0$ with positive probability, in which case we are guaranteed to save evaluating one child.

We modify the algorithm by Saks an Wigderson to tap these opportunities; in each step we query the child which is more likely to evaluate to $0$. Note that this requires knowledge of the distribution. We look at two successive levels of the tree and show that the above considerations bring us significant advantage over the algorithm by Saks and Wigderson.
\paragraph*{Lemma~\ref{lem:main22}}
As mentioned before, here we work with the original definition of sabotage complexity. Our proof splits into the following steps.\begin{enumerate}
\item We recursively define a `hard' distribution $\Prb_d$ over pairs in $g_d^{-1}(0) \times g_d^{-1}(1)$.
\item We consider an arbitrary zero-error randomized algorithm $\mathcal{R}$ for $g_d$. We now wish to give a lower bound on the number of queries it makes on expectation to distinguish a random pair sampled from $\Prb_d$.
\item Using $\mathcal{R}$, we recursively define a sequence of algorithms $\A_d, \A_{d-1}, \ldots, \A_0$ such that for each $i$, $\A_i$ is a zero-error algorithm for $g_i$.
\item We establish a recursive relation amongst the expected number of queries that $g_i$ makes to distinguish a pair sampled from $\Prb_i$, for various $i$. We make a distinction between queries based on their answers ($0$ or $1$). This step involves a small case analysis involving all deterministic trees with two variables.
\item We finish by solving the recursion established in the previous step.
\end{enumerate}
\section{Preliminaries}
\label{sec:prelims}
The notation $[n]$ denotes the set $\{1, \ldots, n\}$. Throughput the paper, $g:\bin^m \to \bin$ will stand for a Boolean function and $x=(x_1,\ldots, x_m)$ will stand for a generic input to $g$. For $b \in \bin$, $f^{-1}(b)=\{x \in \bin^n \mid f(x)=b\}$. For a subset $S$ of $\bin^m$, let $f\mid_S$ denote the restriction of $f$ to $S$.
A probability distribution $\mu$ over $\bin^m$ is a function $\mu:\bin^m \to [0,1]$ such that $\sum_{x \in \bin^m} \mu(x)=1$. For a subset $S$ of $\bin^m$, define $\mu(S):=\sum_{x\in S} \mu(x)$. For a subset $S$ of $\bin^m$ such that $\mu(S)>0$, $\mu|_S$ is the distribution obtained by conditioning $\mu$ on the event that the sample belongs to $S$. In other words:
\[\mu|_S(x)=\left\{\begin{array}{ll}0 & \mbox{if } x \notin S, \\ \frac{\mu(x)}{\mu(S)} & \mbox{if } x \in S\end{array}\right.\]$\mu$ is said to be a product distribution if there exist $p_1, \ldots, p_m \in [0, 1]$ such that for each $x \in \bin^n$, $\mu(x)=\prod_{i=1}^m (x_ip_i + (1-x_i)(1-p_i))$. In other words, each $x_i$ is independently equal to $1$ with probability $p_i$ and $0$ with probability $1-p_i$. Let \textsf{PROD} be the set of all product distributions of $\bin^m$. 

For a subset $I \in [m]$ of indices, $x^{\oplus I}$ denotes the string obtained from $x$ by flipping the variables $x_i$ for each $i \in I$. If $I=\{i\}$, we abuse notation and write $x^{\oplus i}$.
\begin{definition}[Subcube]
\label{def:subcube}
A subset $C$ of $\bin^m$ is called a subcube if there exists a set $S \in [m]$ of indices and bits $\{b_i \mid i \in S\}$ such that $C=\{x \in \bin^m\mid \forall i \in S, x_i=b_i\}$. The co-dimension of $C$ is defined to be $|S|$.
\end{definition}
\subsection{Decision trees for Boolean functions}
\label{sec:decitree}
A decision tree for $m$ variables is a binary tree $T$. Each internal node of $T$ is labelled by a variable $x_i$ for $i \in [m]$, and has two children that corresponds to $x_i=0$ and $x_i=1$. Each leaf is labelled by $0$ or $1$. A decision tree is evaluated on a given input $x=(x_1,\ldots,x_m)$, as follows.  Start at the root. In each step, if the current node is an internal node, then query its label $x_i$. Then navigate to that child of the current node that corresponds to the value of $x_i$. The computation stops when it reaches a leaf, and outputs the label of the leaf. Let $T(x)$ denotes the output of the tree at $x$.

The inputs $x$ that take the tree $T$ to leaf $\ell$ is exactly the ones which agree with the path from the root to $\ell$ for every variable queried on the path. Thus, the set of such inputs is a subcube of $\bin^m$ of co-dimension equal to the depth of $\ell$. The notation $\ell$ will also refer to the subcube corresponding to the leaf $\ell$.  

$T$ is said to compute $g:\bin^m \to \bin$ if 
\[\forall x \in \bin^m, T(x)=g(x).\]The \emph{Deterministic Decision Tree complexity of $g$}, denoted by $\D(g)$ is the minimum depth of a decision tree that computes $f$.

Let $\mu$ be a distribution over $\bin^m$. For a given error parameter $\epsilon \in [0, 1/2]$, $T$ computes $g$ with error probability $\epsilon$ over $\mu$ if
\[\Pr_{x \sim \mu}[g(x)\neq T(x)]\leq \epsilon.\]
The \emph{distributional query complexity of $g$ for error $\epsilon$ with respect to $\mu$}, denoted by $\D^\mu_\epsilon(f)$, is the minimum depth of a decision tree that computes $f$ with error probability $\epsilon$ over $\mu$.

A randomized decision tree is a probability distribution $\mathcal{R}$ over deterministic decision trees. $\mathcal{R}$ is said to compute $g$ with error probability $\epsilon$ if
\[\forall x \in \bin^m, \Pr_{T \sim \mathcal{R}}[T(x) \neq g(x)]\leq \epsilon.\]
The query complexity of $\mathcal{R}$ is the maximum depth of any decision tree in its support. The r\emph{randomized query complexity of $g$ for error $\epsilon$}, denoted by $\R_\epsilon(g)$, is the minimum query complexity of any randomized decision tree $\mathcal{R}$ that computes $g$ with error $\epsilon$. Define $\R(g):=\R_{1/3}(g)$.
The following fact is well-known (see, for example \cite{GLSS19} for a proof).
\begin{fact}[Minimax theorem]
\label{fact:minimax}
$\mathcal{R}_\epsilon(g)=\max_{\mu}\D^\mu_\epsilon(g).$
\end{fact}
We define the \emph{product distributional query complexity of $g$} with error $\epsilon$, $\D^\product_\epsilon(g)$, as follows. \[\D^\product_\epsilon(g):=\max_{\mu \in \mathsf{PROD}} \D^\mu_\epsilon(g).\]
\subsection{Sensitivity and influence}
\label{sec:influence}
For a total Boolean function $g$, a variable $x_i$ is said to be sensitive for an input $x$ if $g(x)\neq g(x^{\oplus i})$. The sensitivity of $x$ with respect to $g$, denoted by $s(g,x)$, is the number of sensitive bits of $x$, i.e., $|\{i \in [n]\mid g(x)\neq g(x^{\oplus i})\}|$. The sensitivity of $g$, denoted by $s(g)$, is the maximum sensitivity of any input $x$ with respect to $g$, i.e.,
\[s(g)=\max_{x \in \bin^m}s(g,x).\]
For a product distribution $\mu \in \mathsf{PROD}$ given by parameters $p_1,\ldots,p_m$, the \emph{influence} of $x_i$ with respect to $g$ and $\mu$ is defined as
\[\I_i(g):=4p_i(1-p_i)\Pr_{x \sim \mu}[g(x)\neq g(x^{\oplus i})],\]
and the influence of $g$ with respect to $\mu$ is defined as
\[\I(g)=\sum_{i=1}^m \I_i(g).\]
The following inequality follows easily from the above definitions, linearity of expectation, and the observation that $4p_i(1-p_i) \leq 1$ for all $p_i \in [0, 1]$.
\begin{align}
\label{eq:inf_avg_s}
\I(g) \leq \E_{x \sim \mu} s(g, x).
\end{align}
Let $\Var(g)$ denote the variance of the random variable $g(x)$ when $x$ is drawn from $\mu$. The Poincar\'{e} inequality bounds $\Var(g)$ in terms of $\I(g)$.
\begin{lemma}[Poincar\'{e} inequality]
\label{lem:poincare}
For every product distribution $\mu$, $4\Var(g) \leq \I(g)$.
\end{lemma}
A proof of the Poincar\'{e} inequality may be found in \cite{o14}.

\noindent In the notations $\I_i, \I$ and $\Var$, the dependence on $\mu$ is suppressed. $\mu$ will be clear from the context.
\subsection{Randomized query complexity for product distributions}
\label{par:rprod}
Let $\mu$ be a product distribution, and $T$ be a deterministic decision tree. For each leaf $\ell$ of $T$, let $p^\mu_\ell$ be the probability that the computation of $T$ on an input drawn from $\mu$ reaches $\ell$. Let $\mathbf{p}^\mu$ denote the probability distribution $(p^\mu_\ell)_{\ell}$ over the leaves of $T$. We say that a randomized decision tree $\mathcal{R}$ computes $g$ with error $\epsilon$ for product distributions if for every product distribution $\mu \in \mathsf{PROD}$,
\[\E_{T \sim \mathcal{R}}\E_{\ell\sim \mathbf{p}^\mu} [\min\{\Pr_{x \in \mu|_\ell}[g(x)=0], \Pr_{x \in \mu|_\ell}[g(x)=1]\}] \leq \epsilon,\]
where the inner expectation is over the leaves of $T$. We define $\min\{\Pr_{x \in \mu\mid_\ell}[f(x)=0],\Pr_{x \in \mu\mid_\ell}[f(x)=1]\}$ to be $0$ if $p_\ell^\mu=0$; the conditional distribution $\mu|_\ell$ is not defined in this case. The \emph{randomized query complexity} of $g$ for product distribution for error $\epsilon$, denoted by $\R^{\product}_\epsilon(f)$, is the minimum query complexity of a randomized decision tree $\mathcal{R}$ that satisfies the above condition. Define $\R^{\product}(f):=\R^{\product}_{1/3}(f)$.

 Note that in the above definition, no reference has been made to the labels of the leaves of $T$. For the purpose of this definition, $\mathcal{R}$ can be thought of as a probability distribution over trees whose leaves are unlabelled.

 The following claim shows that $\D^\product_\epsilon(g)$ is bounded above by $\R^{\product}_\epsilon(g)$. A proof may be found in Section~\ref{sec:dprod<=rprod}.
 
 % For every product distribution $\mu$, each leaf $\ell$ of each $T$ in the support of $\mathcal{R}$ can be labelled by the bit which has at least $1/2$ probability mass in the leaf with respect to $\mu|_\ell$ (thus the leaf labels will depend on the distribution). The resulting randomized decision tree is seen to decide $g$ with error probability at most $\epsilon$ with respect to $\mu$. By fixing the randomness of $\mathcal{R}$ appropriately (depending on $\mu$) we will get a deterministic decision tree that decides $g$ with error probability at most $\epsilon$ with respect to $\mu$. We thus have the following.
\begin{claim}
\label{clm:rprod>=dprod}
For every Boolean function $g$ and parameter $\epsilon \in [0, 1/2]$,
$\D^\product_\epsilon(g) \leq \R^{\product}_\epsilon(g)$.
\end{claim}
\subsection{Sabotage complexity.}
\label{par:sabotage}
The \emph{sabotage complexity} of a Boolean function $g$ for error $\epsilon$, denoted by $\RS_{\epsilon}(g)$, is defined to be the minimum query complexity of any randomized decision tree $\mathcal{R}$ for which the following is true: For every $x = (x_1,\ldots, x_m)\in g^{-1}(0)$, $y=(y_1,\ldots, y_m) \in g^{-1}(1)$, with probability at least $1-\epsilon$, a decision tree $T$ drawn from $\mathcal{R}$ when run on $x$ queries an index $i$ such that $x_i\neq y_i$\footnote{Note that $T$ queries an index $i$ such that $x_i\neq y_i$ when run on $x$ if and only if $T$ queries an index $j$ such that $x_j\neq y_j$ when run on $y$.}. Define $\RS(g):=\RS_{1/3}(g)$.

Sabotage complexity was defined by Ben-David and Kothari \cite{BK15}. They defined the measure as the minimum expected query complexity of any randomized decision tree to distinguish each pair of inputs $x \in g^{-1}(0), y \in g^{-1}(1)$. However, as the authors observed, the definition stated above is within a constant factor of the original definition in \cite{BK15}. See more discussion on this in Section~\ref{sec:rs-lb-nand} where we work with the original definition.

The following fact can be proven by standard BPP amplification.
\begin{fact}
\label{fact:error_red}
$\forall \epsilon, \epsilon' \in (0,1)$ and $\epsilon<\epsilon'$, $\RS_\epsilon(g)=O\left(\RS_{\epsilon'}(g)\cdot\frac{\log(1/\epsilon)}{\log(1/\epsilon')}\right)$.
\end{fact}
Ben-David and Kothari proved that the sabotage complexity is lower bounded by many complexity measures that are studied in the context of decision trees. In particular, $\RS(g)$ is lower bounded by $s(g)$.
\begin{fact}[\cite{BK15}]
\label{fact:RS>s}
For all Boolean function $g:\bin^m \to \bin$, $\RS(g)=\Omega(s(g))$.
\end{fact}
\section{Sabotage complexity and product distributions}
In this section we first prove Lemma~\ref{lem:main1}. We then use Lemma~\ref{lem:main1} to prove Theorems~\ref{thm:th2} and~\ref{thm:th3}.
\ \\
The following lemma says that to distinguish each pair of inputs on which the function values differ with high probability, it is necessary and sufficient to query each sensitive bit of each input with high probability.
\begin{lemma}
\label{lem:rs_and_sens_ip}
Let $g:\bin^n \to \bin$ be a total Boolean function. Then, $\RS_\epsilon(g)\leq r$ if and only if there is a randomized decision tree $\mathcal{R}$ of query complexity at most $r$ such that for each input $x$ and each variable $x_i$ sensitive for $x$, $\Pr_{T \sim \mathcal{R}}[\mbox{$T$ does not query $x_i$ when run on $x$}] \leq \epsilon$.
\end{lemma}
\begin{proof}
\begin{description}
\item[(If)] Let $\mathcal{R}$ be a randomized decision tree of complexity at most $r$ such that for every input $x$ and every variable $x_i$ sensitive for $x$, $\Pr_{T \sim \mathcal{R}}[\mbox{$T$ does not query $x_i$ when run on $x$}] \leq \epsilon$. We will show that $\mathcal{R}$ fails to distinguish any pair $w \in g^{-1}(0), y \in g^{-1}(1)$ with probability at most $\epsilon$. Fix such a pair $w, y$. let $B=\{i_1,\ldots, i_k\}$ be the positions where $w$ and $y$ differ. Define $B_0:=\emptyset$ and for $1 \leq j \leq k$, define $B_j:=\{i_1,\ldots, i_j\}$. Let $m$ be the smallest index such that $g(w^{\oplus B_m})=1$. Thus, variable $w_m$ is sensitive for $w^{\oplus B_{m-1}}$ and $w^{\oplus B_m}$. Now, observe that if $T$ does not query any variable $w_{i_j}$ with $i_j \in B$ when run on $w$, then $T$ does not query $w_m$ when run on $w^{\oplus B_{m-1}}$. By our assumption about $\mathcal R$, the probability of this happening when $T$ is sampled from $\mathcal R$ is at most $\epsilon$.
\item[(Only if)] Let $\RS_\epsilon(g)\le r$. Thus there exists a randomized decision tree $\mathcal{R}$ of query complexity $r$ that fails to distinguish each pair $w \in g^{-1}(0), y \in g^{-1}(1)$ with probability at most $\epsilon$. Without loss of generality, assume that $x \in g^{-1}(0)$. Then $x^{\oplus i} \in g^{-1}(1)$. Since distinguishing $x$ and $x^{\oplus i}$ is equivalent to querying $x_i$ when run on $x$, the proof follows.
\end{description}
\end{proof}
Now we proceed to proving Lemma~\ref{lem:main1}. For convenience, we use $\epsilon$ for $\epsilon(n)$ throughout the following proof.
\begin{proof}[Proof of Lemma~\ref{lem:main1}]
\begin{description}
\item[(Part 1)] Let $\RS_{1-\epsilon}(g)=r$. By Lemma~\ref{lem:rs_and_sens_ip}, there exists a randomized query algorithm $\mathcal{R}$ of complexity at most $r$ such that for each input $x$ and each variable $x_i$ sensitive for $x$, $\Pr_{T \sim \mathcal{R}}[\mbox{$T$  does  not query $x_i$ when run on $x$}] \leq 1-\epsilon$. Let $\mathcal R'$ be the algorithm obtained by repeating $\mathcal R$ $\frac{2}{\epsilon}\ln s(g)$ times with independent randomness. Thus for each input $x$ and each variable $x_i$ sensitive for $x$, we have that $\Pr_{T \sim \mathcal{R}'}[\mbox{$T$ does not query $x_i$ when run on x}] \leq (1-\epsilon)^{(\frac{1}{\epsilon}\cdot 2 \ln s(g))} \leq \frac{1}{s(g)^2}$, where we have used the inequality $1-x \leq e^{-x}$ for all $x \in (-\infty, \infty)$. Again for each input $x$, by a union bound over all variables $x_i$ sensitive for $x$, we have that the probability that a deterministic tree sampled from $\mathcal{R}'$ does not query all variables sensitive for $x$ when run on $x$, is at most $\frac{s(g,x)}{s(g)^2}\leq \frac{1}{s(g)}$. The query complexity of $\mathcal{R}'$ is $O(\frac{1}{\epsilon}\cdot \RS_{1-\epsilon}(g) \log s(g))$. We will show that $\mathcal{R}'$ computes $g$ with error $1/3$ for product distributions. This will complete the proof of this part.

To this end, fix a product distribution $\mu$. For any deterministic decision tree $T$ and input $x$  of $f$, define
\[\Qu(T, x)=\left\{\begin{array}{ll}1 & \mbox{if $T$ does not query all sensitive variables of $x$ when run on $x$,}\\0 & \mbox{otherwise.}\end{array}\right.\]
By the property of $\mathcal{R}'$, for every input $x$, we have that
\[\E_{T \sim \mathcal{R}'}[\Qu(T, x)]=\Pr_{T \sim \mathcal{R}'}[\Qu(T, x)=1]\leq \frac{1}{s(g)}.\]
Since the above is true for each $x$, we have the following for a random input $x$ sampled from $\mu$.
\begin{align}
\label{eq:mainlemeq1}\E_{T \sim \mathcal{R}'}\E_{x \sim \mu}[\Qu(T, x)]\leq \frac{1}{s(g)}.\end{align}
For each leaf $\ell$ of $T$ , let $p^\mu_\ell$ be the probability that the computation of $T$ on an input drawn from $\mu$ reaches $\ell$ and $\mathbf{p}^\mu$ denote the probability distribution $(p^\mu_\ell)_{\ell}$ over the leaves of $T$. We rewrite (\ref{eq:mainlemeq1}) as follows.
\begin{align}
\label{eq:mainlemeq2}\E_{T \sim \mathcal{R}'}\E_{\ell\sim p^\mu}\E_{x \sim \mu\mid_\ell}[\Qu(T, x)]\leq \frac{1}{s(g)},\end{align}
treating $\E_{x \sim \mu\mid_\ell}[\Qu(T, x)]$ as $0$ if $p^\mu_\ell=0$. Now, fix an arbitrary leaf $\ell$ of $T$ such that $p^\mu_\ell>0$, and consider the Boolean function $g\mid_\ell$. Note that for any $x \in \ell$, if $\Qu(T, x)=0$, then $s(g\mid_\ell, x)=0$. We thus have that
\begin{align}
\label{eq:maianlemeq3}
\E_{x\sim\mu\mid_\ell}[s(g\mid_\ell, x)]\leq \Pr_{x \sim \mu\mid_\ell}[\Qu(T, x)=1]\cdot s(g\mid_\ell) \leq \E_{x \sim \mu\mid_\ell}[\Qu(T, x)]\cdot s(g).
\end{align}
Since $\mu$ is a product distribution and $\ell$ is a subcube, $\mu\mid_\ell$ is also a product distribution. Equations~(\ref{eq:inf_avg_s}) and~(\ref{eq:maianlemeq3}) thus imply that
\begin{align}
\label{eq:mainlemeq4}
\I(g\mid_\ell) \leq \E_{x \sim \mu\mid_\ell}[\Qu(T, x)]\cdot s(g).
\end{align}
Together with Poincar\'{e} inequality (Lemma~\ref{lem:poincare}), (\ref{eq:mainlemeq4}) implies that \begin{align}
\label{eq:mainlemeq5}
\Var(f\mid_\ell) \leq \frac{1}{4}\cdot\E_{x \sim \mu\mid_\ell}[\Qu(T, x)]\cdot s(f).
\end{align}
 Now, for a random variable $X$ taking value in $\{0, 1\}$, $\Var(X)=4\Pr[X=0]\Pr[X=1]\geq 2 \min\{\Pr[X=0], \Pr[X=1]\}$ (since $\max\{\Pr[X=0], \Pr[X=1]\}\geq \frac{1}{2}$). Since $\ell$ is an arbitrary leaf, we have by Equations (\ref{eq:mainlemeq5}) and (\ref{eq:mainlemeq2}) that
\begin{align*}
&\E_{T \sim \mathcal{R}'}\E_{\ell \sim p^\mu}[\min\{\Pr_{x\sim \mu\mid_\ell}[g(x)=0], \Pr_{x\sim \mu\mid_\ell}[g(x)=1]\}] \\
\leq &\frac{1}{2} \cdot \E_{T \sim \mathcal{R}'}\E_{\ell \sim p^\mu}[\Var(g\mid_\ell)] &\text{by the above discussion} \\
\leq &\frac{1}{8} \cdot \E_{T \sim \mathcal{R}'}\E_{\ell \sim p^\mu}\E_{x \sim \mu\mid_\ell}[\Qu(T, x)]\cdot s(g) & \text{by Equation~(\ref{eq:mainlemeq5})} \\
\leq & \frac{1}{8}<\frac{1}{3}. & \text{by Equation~(\ref{eq:mainlemeq2})}
\end{align*}
Since $\mu$ is an arbitrary product distribution, we have that $\mathcal{R}'$ computes $g$ with error $1/3$ for product distributions.
\item[(Part 2)] Fix a randomized query algorithm $\mathcal{R}$ that attains $\R^\product_{\frac{1}{2}-\epsilon}(g)$. We will show that $\mathcal{R}$ also attains $\RS_{1-2\epsilon}(g)$. By Lemma~\ref{lem:rs_and_sens_ip} it is sufficient to show that for each input $x$ and each variable $x_i$ sensitive for $x$, $\Pr_{T \sim \mathcal{R}}$[$T$ does not query $x_i$ when run on $x$]$\leq 1-2\epsilon$. To this end, fix an input $x$ and a variable $x_i$ sensitive for $x$. Now consider the distribution $\mu$ that places a probability mass of $1/2$ on $x$ and places the remaining mass of $1/2$ on $x^{\oplus i}$. Note that $\mu$ is a product distribution. Thus from the property of $\R$ we have that
\begin{align}
\label{eq:part2-1}
    \E_{T \sim \mathcal{R}}\E_{\ell\sim \mathbf{p}^\mu} [\min\{\Pr_{x \in \mu|_\ell}[g(x)=0], \Pr_{x \in \mu|_\ell}[g(x)=1]\}] \leq \frac{1}{2}-\epsilon.
\end{align}
Now if $T$ does not query $x_i$ when run on $x$, then $T$ has a leaf $\ell$ that contains both $x$ and $x^i$, $p^\mu_\ell=1$, and for all other leaves $\ell'$ of $T$, $p^\mu_{\ell'}=0$. Furthermore, $\min\{\Pr_{x \in \mu|_\ell}[g(x)=0], \Pr_{x \in \mu|_\ell}[g(x)=1]\}=1/2$. Thus, $\E_{\ell\sim \mathbf{p}^\mu} [\min\{\Pr_{x \in \mu|_\ell}[g(x)=0], \Pr_{x \in \mu|_\ell}[g(x)=1]\}] =1/2$. We thus have that,
\begin{align}
&\E_{T \sim \mathcal{R}}\E_{\ell\sim \mathbf{p}^\mu} [\min\{\Pr_{x \in \mu|_\ell}[g(x)=0], \Pr_{x \in \mu|_\ell}[g(x)=1]\}]\nonumber \\
&\geq \Pr_{T\sim\mathcal{R}}[\text{$T$ does not query $x_i$ when run on $x$}]\cdot\frac{1}{2}. \label{eq:part2-2}
\end{align}
Equations~(\ref{eq:part2-1}) and~(\ref{eq:part2-2}) imply that
\begin{align*}
\Pr_{T\sim\mathcal{R}}[\text{$T$ does not query $x_i$ when run on $x$}]\cdot\frac{1}{2} &\leq \frac{1}{2}-\epsilon \\
\implies \Pr_{T\sim\mathcal{R}}[\text{$T$ does not query $x_i$ when run on $x$}] &\leq 1-2\epsilon.
\end{align*}
This completes the proof.
\end{description}
\end{proof}
Now we prove Theorems~\ref{thm:th2} and~\ref{thm:th3}.
\begin{proof}[Proof of Theorem~\ref{thm:th2} Part 1]
Substituting $\epsilon(n)=1/6$ in part 2 of Lemma~\ref{lem:main1} we have that
\begin{align}
\R^\product(g)&\geq \RS_{2/3}=\Omega(\RS(g)),\nonumber
\end{align}
where the second equality follows from Fact~\ref{fact:error_red}.
\end{proof}
Theorem~\ref{thm:th2} (1) and Fact~\ref{fact:RS>s} imply that
\begin{align}
\Rprodnb{}(g)=\Omega(s(g)). \label{eq:rprods}
\end{align}
\begin{proof}[Proof of Theorem~\ref{thm:th2} Part 2]
\begin{align}
\RS(g) = \RS_{1/3}(g) &\geq \RS_{2/3}(g) \nonumber \\
&=\Omega(\Rprodnb{}(g)/\log s(g)) & \text{by Lemma~\ref{lem:main1} part 1 with $\epsilon(n)=1/3$} \nonumber \\
&=\Omega(\Rprodnb{}(g)/\log \Rprodnb{}(g)) &\text{by Equation~(\ref{eq:rprods})} \nonumber
\end{align}
\end{proof}
\begin{proof}[Proof of Theorem~\ref{thm:th3}]
We have
\begin{align}
\R^\product(g) &= O\left(\frac{1}{\epsilon(n)}\cdot \RS_{1-\epsilon(n)}(g) \log s(g)\right) &\text{Part (1) of Lemma~\ref{lem:main1}}\nonumber\\
&=O\left(\frac{1}{\epsilon(n)}\cdot \RS_{1-2\epsilon(n)}(g) \log s(g)\right) &\text{since $1-\epsilon(n) \geq 1-2\epsilon(n)$}\nonumber\\
&=O
\left(\frac{1}{\epsilon(n)}\cdot \R^\product_{\frac{1}{2}-\epsilon(n)}(g) \log s(g)\right) &\text{by part (2) of Lemma~\ref{lem:main1}} \nonumber \\
&=O\left(\frac{1}{\epsilon(n)}\cdot \R^\product_{\frac{1}{2}-\epsilon(n)}(g))\log \R^\product(g)\right). &\text{by Equation~(\ref{eq:rprods})}
\nonumber
\end{align}
\end{proof}
\section{Separation between $\D^\product$ and $\R$}
\label{sec:nand_dist}
In this section we prove Lemma~\ref{lem:main21}. Recall that $g_d$ denotes the NAND tree function of depth $d$. Snir \cite{S85} and Saks and Wigderson \cite{SW86} were the first to study $g_d$ in the context of randomized query complexity. As mentioned in Section~\ref{sec:intro}, it is known from the works of Saks and Wigderson \cite{SW86} and Santha \cite{S91} that $\R(g_d)=\Theta(d^\alpha$) where $\alpha =\frac{1+\sqrt{33}}{4}$.

For a distribution $\mu$, the the zero-error distributional complexity of a Boolean function $g$, that we denote by $\overline{\D_0^\mu(g)}$, is the least expected number of queries made by any (deterministic) tree $T$ on a random input sampled from $\mu$. Define $\Dprod{0}(g):=\max_{\mu \in \mathsf{PROD}}\overline{\D_0^\mu(g)}$. By Markov's inequality, it follows that $\Dprodnb{}(g)=O(\Dprod{0}(g))$.
\begin{proof}[Proof of Lemma~\ref{lem:main22}]
We will prove an upper bound on $\Dprod{0}(g)$. By the preceding discussion, that will prove the lemma.

Let $\mu$ be any product distribution over $\bin^n$. Define $T(d, \mu):=\Dprodmu{0}{\mu}(g_d)$ and $T(d):=\Dprod{0}(g_d)$. Consider the query algorithm $\A_\mu$\footnote{Note that $\A_\mu$ needs the knowwledge of $\mu$.} given in Algorithm~\ref{algo:nand-distn}.\\

\begin{algorithm}[H]
\label{algo:nand-distn}
\SetAlgoLined
\textbf{Input: } Query access to $x=(x_1, \ldots, x_{2^d})$.\\
$g \gets g_d$.\\
\If{$g$ is a variable}{Query $g$. Return the outcome of the query.}
\Else{
    Let $g_\ell$ and $g_r$ respectively be the left and right subtrees of $g$. \\
    Let $\mu_\ell$ and $\mu_r$ respectively be the product distributions induced on the input spaces of $g_\ell$ and $g_r$ by $\mu$.\\
    $t\gets \arg \max_{i \in \{\ell, r\}} \Pr_{y \sim \mu_i}[g_i(y)=0]$. \\
    $s \gets \{\ell, r\}\setminus \{t\}$. \\ For $i \in \{\ell, r\}$, let $x^{(i)}$ be the input to $g_i$. \\
    \If{$\A_{\mu_t}(x^{(t)})=0$}{return $1$.}
    \Else{return 
    $\overline{\A_{\mu_s}(x^{(s)})}$.}
}
\caption{$\A_\mu(x)$}
\end{algorithm}
$\A_\mu$ works as follows: if $d=1$, i.e., if $g_d$ is a single variable, then $\A_\mu$ queries and returns the value of the variable. Else, $\A_\mu$ recursively evaluates a subtree of the root of $g_d$ whose probability of evaluating to $0$ is at least that of the other subtree.\footnote{By `recursively evaluates' we  mean that $\A_\mu$ invokes $\A_{\mu'}$ for the distribution $\mu'$ induced by $\mu$ on the domain of the subfunction under consideration.} If the recursive call returns $0$, $\A_\mu$ returns $1$. Else, $\A_\mu$ recursively evaluates the other subtree of the root of $g_d$ and returns the complement of the value returned by that recursive call. It is clear that on every input, $\A_\mu$ returns the correct answer with probability $1$.

Now we analyze the query complexity of $\A$. For $i \in \{\ell, r\}$, define $p_i:=\Pr_{x^{(i)} \sim \mu_i}[g_d(x^{(i)})=0]$. WLOG assume that $p_\ell \ge p_r$. $\A_\mu$ on input $x$ will recursively evaluate $g_\ell$ by invoking $\A_{\mu_\ell}$ on input $x^{(\ell)}$. If the recursive call returns $1$, then $\A$ will recursively evaluate $g_r$ by invoking $\A_{\mu_r}$ on input $x^{(r)}$. We thus have that,
\begin{align}
T(d, \mu)=T(d-1, \mu_\ell)+(1-p_\ell)T(d-1, \mu_r) .\label{eq:nandproeq1}
\end{align}
Let $\alpha_\ell, \alpha_r$ respectively be the probabilities that the left and right children of $g_\ell$ evaluate to $0$. Similarly, let $\beta_\ell, \beta_r$ respectively be the probabilities that the left and right children of $g_r$ evaluate to $0$. Without loss of generality assume that $\alpha_\ell\geq \alpha_r, \beta_\ell \ge \beta_r$ and $\alpha_\ell \le \beta_\ell$ (other cases are similar). By using a similar analysis as above and then upper bounding distributional query complexity for specific product distributions by the product distributional complexity we have that,
\begin{align}
T(d-1, \mu_\ell) &\leq T(d-2)+(1-\alpha_\ell)T(d-2),\mbox{\ and}\label{eq:eq2} \\
T(d-1, \mu_r) &\leq T(d-2)+(1-\beta_\ell)T(d-2)
 \nonumber \\
 &\leq T(d-2)+(1-\alpha_\ell)T(d-2). \label{eq:eq3} \end{align}
 Substituting Equations~(\ref{eq:eq2}) and~(\ref{eq:eq3}) in~(\ref{eq:nandproeq1}) we have that
 \begin{align}
T(d)&\leq T(d, \mu) \nonumber \\
&\leq (2-\alpha_\ell)(2-p_\ell)T(d-2) \label{eq:eq4}.
 \end{align}
Now, we have that $p_\ell=(1-\alpha_r)(1-\alpha_\ell) \geq (1-\alpha_\ell)^2$.\footnote{Here we use that $\mu$ is a product distribution.} Substituting in Equation~(\ref{eq:eq4}) we have that
\begin{align}
T(d) &\leq (2-\alpha_\ell)(2-(1-\alpha_\ell)^2)T(d-2) \nonumber \\
&=(2-\alpha_\ell)(1+2\alpha_\ell-\alpha_\ell^2)T(d-2) \label{eq:eq5}.
\end{align}
The maximum value of the function $f(x):=(2-x)(1+2x-x^2)$ in the domain $[0,1]$ is $\frac{2}{27}\cdot(17+7\sqrt 7)$. From Equation~(\ref{eq:eq5}) we have that
\[T(d)=O\left(\sqrt{\frac{2}{27}\cdot(17+7\sqrt 7))}\right)^d=O(\alpha-\delta)^d\mbox{\ \ for some constant $\delta >0$.}\]\end{proof}
\section{Sabotage complexity of NAND tree}
\label{sec:rs-lb-nand}
In this section, we prove Lemma~\ref{lem:main22}. Recall  that $g_d$ stands for the NAND tree function of depth $d$. Define $g_0(b)=b$ for $b \in \bin$.

For a randomized query algorithm $\mathcal{R}$ that decides $g:\bin^m \to \bin$ with error probability $0$, and for inputs $x, y$ such that $g(x)=0, g(y)=1$, define the expected sabotage complexity of $\mathcal{R}$ on the pair $x, y$, denoted by $\RS_E(\mathcal{R}, x, y)$, to be the expected number of queries that $\mathcal{R}$ makes until (and including) it queries an index $i$ such that $x_i \neq y_i$ when run on $x$ (or $y$). Define the expected sabotage complexity $\RS_E(\mathcal{R})$ to be $\max_{x, y \in \bin^m, g(x)=0, g(y)=1} \RS_E(\mathcal{R}, x, y)$, and the expected sabotage complexity $\RS_E(g)$ to be the minimum $\RS_E(\mathcal{R})$ for any randomized query algorithm $\mathcal{R}$ that decided $g$ with error probability $0$. As observed by Ben-David and Kothari, $\RS_E(g)=\Theta(\RS(g))$. In this section, we will work with $\RS_E$ in place of $\RS(g)$.

It follows by standard arguments that for every distribution $\mathcal{D}$ on $g^{-1}(0) \times g^{-1}(1)$ there exists a zero-error randomized (even deterministic) decision tree $\mathcal{R}$ of $g$ such that $\E_{(x, y)\sim \mathcal{D}}[\RS_E(|R, x, y)] \leq \RS_E(g)$. To prove Lemma~\ref{lem:main22} it thus suffices to exhibit a hard distribution $\mathcal{D}$ on $g^{-1}(0) \times g^{-1}(1)$ such that for every zero-error randomized tree $\mathcal{R}$ of $g$, $\E_{(x, y)\sim \mathcal{D}}[\RS_E(|\mathcal{R}, x, y)]$ is large. The first step in our proof is to define a hard distribution.
\subsection*{A hard distribution}
We define a probability distribution $\Prb_d$ on ${g_d}^{-1}(0) \times {g_d}^{-1}(1)$ as follows. Define $\Prb_0$ to be the point distribution $\{(0, 1)\}$. For $d \ge 1$, $\Prb_d$ is defined recursively by the following sampling procedure. Let $n:=2^{d-1}$.  \ \\
\begin{enumerate}
\item Sample $(x, y) \sim \Prb_{d-1}$. Let $x=(x_1, \ldots, x_n)$ and $y=(y_1, \ldots, y_n)$.
\item Sample $\overline{b}:=(b_1,\ldots, b_n)$ uniformly at random from $\bin^n$.
\item For each $i=1, \ldots, n$, let $u_i=(u_i^{(0)}, u_i^{(1)}), v_i=(v_i^{(0)}, v_i^{(1)}) \in \bin^2$ be defined as follows:
\begin{enumerate}
\item If $(x_i, y_i)=(0, 0)$, set $u_i, v_i \gets (1, 1)$.
\item If $(x_i, y_i)=(0, 1)$, set $u_i \gets (1, 1)$ and set $v_i \gets (b_i, 1-b_i)$.
\item If $(x_i, y_i)=(1, 0)$, set $u_i \gets (b_i, 1-b_i)$ and set $v_i \gets (1, 1)$.
\item If $(x_i, y_i)=(1, 1)$, set $u_i, v_i \gets (b_i, 1-b_i)$.
\end{enumerate}
\item Let $x'$ be the string obtained from $x$ by replacing each $x_i$ by $u_i$. Similarly let $y'$ be the string obtained from $y$ by replacing each $y_i$ by $v_i$.
\item Return $(x', y')$.
\end{enumerate}
\ \\
Notice that for each $i=1, \ldots, n$, $x_i=\mathsf{NAND}(u_i^{(1)}, u_i^{(2)})$ and $y_i=\mathsf{NAND}(v_i^{(1)}, v_i^{(2)})$. Hence, $g_d(x')=g_{d-1}(x)$ and $g_d(y')=g_{d-1}(y)$. Thus we inductively establish that $\Prb_d$ is supported on $g_d^{-1}(0) \times g_d^{-1}(1)$. The following observation can be verified to be true by a simple case analysis.
\begin{observation}
\label{obs:embed}
For each $i=1, \ldots, n$, $x_i=\overline{u_{i}^{(b_i)}}$ and $y_i=\overline{v_{i}^{(b_i)}}$. Furthermore, $u_{i}^{(1-b_i)}=v_{i}^{(1-b_i)}=1$.
\end{observation}
In light of Observation~\ref{obs:embed}, the sampling process above can be intuitively described as follows. We first sample $(x, y)$ from $\Prb_{d-1}$. Then, for each $i$, we sample two two-bit strings $u_i$ and $v_i$ that are jointly distributed in a certain way. If $x_i=y_i$, then $u_i=v_i$. If $x_i \neq y_i$, then the values of $x_i$ and $y_i$ are embedded (as complements) in the $b_i$-th bits of $u_i$ and $v_i$ respectively. The $(1-b_i)$-th bit of $u_i$ and $v_i$ are set to $1$. The marginals of $\Prb_d$ can be seen to be obtained by conditioning uniform distribution on the ``reluctant inputs" considered by Saks and Wigderson \cite{SW86} to the events $g(x)=0$ and $g(x)=1$. We couple these two conditional distributions in a specific way to obtain $\Prb_d$.
\subsection*{A sequence of algorithms}
Now we proceed to prove a lower bound on $\E_{(x,y)\sim\Prb_d}[\RS_E(\R,x, y)]$ for any zero-error algorithm $\mathcal{R}$ of $g_d$. Towards this goal, let $\mathcal{R}$ be a zero-error randomized query algorithm for $g_d$. Now, using $\mathcal{R}$, we will define a sequence of randomized query algorithms $\A_d, \A_{d-1}, \ldots, \A_1, \A_0$, where for each $t =d, d-1,\ldots, 0$, $\A_{t}$ is a zero-error randomized query algorithm for $g_{t}$. Define $\A_d:=\mathcal{R}$. Now for $t\leq d-1$, define $\A_t$ recursively as follows. Let $x=(x_1, \ldots, x_{2^t})$ be the input to $\A_t$.  \ \\
\begin{enumerate}
\item Sample $\overline{b}=(b_1, \ldots, b_{2^t})$ uniformly at random from $\bin^{2^t}$.
\item For each $i=1,\ldots, 2^t$, define $u_i \in \bin^2$ as in the definition of $\Prb_d$ above. Let $x' \in \bin^{2^{t+1}}$ be the string obtained from $x$ by replacing each $x_i$ by $u_i$.
\item Simulate $\A_{t+1}$ on $x'$. If $\A_{t+1}$ queries $u_i^{(1-b_i)}$ for some $i$, answer $1$. If $\A_{t+1}$ queries $u_i^{(b_i)}$ for some $i$, make a query to $x_i$ and answer $\overline{x_i}$. The correctness of this simulation follows from Observation~\ref{obs:embed}.
\item When $\A_{t+1}$ terminates, terminate and return what $\A_{t+1}$ returns.\footnote{The return value is not important here. We are bothered only about separating $x$ and $y$. The algorithms may be thought to have unlabelled leaves.}
\end{enumerate}\ \\
We observe that $g_t(x)=g_{t+1}(x')$. Thus we may inductively establish that for every $t=1,\ldots, d$, $\A_t$ is a zero-error randomized decision tree of $g_t$. Moreover, observe that sampling $(x, y)$ from $\Prb_t$ and running $A_{t}$  on $x$ (or $y$) amounts to sampling $(x', y')$ from $\Prb_{t+1}$ and running $A_{t+1}$ on $x'$ (or $y'$). Furthermore, $\A_t$ queries the first index $i$ such that $x_i \neq y_i$ exactly when the simulation of $\A_{t+1}$ inside it queries the first index $j$ such that $x'_j \neq y'_j$. We will index the bits of $x'$ as tuples $(i, b)$ where $i \in \{1,\ldots, 2^t\}$ and $b \in \bin$. Thus $x'_{(i,b)}=u_i^{(b)}$.
\subsection*{The lower bound}
For each $b \in \bin$, each $t=0,\ldots, d$ and each $(x, y)$ in the support of $\Prb_t$, define $Q(t, b, x, y)$ to be the number of variables with value $b$ queried by $\A_t$ when run on $x$ until (and including) $\A_t$ queries an index $i$ such that $x_i \neq y_i$. Define $Q(t, b)$ to be the expected value of $Q(t, b, x, y)$ for a random sample $(x, y)$ from $\Prb_t$, where the expectation is over both the internal randomness of $\A_t$, and the randomness of $\Prb_t$. Our goal is to derive a recursive relationship amongst the quantities $Q(t, b)$, and then obtain a lower bound on $Q(d, b)$. Since $\E[\RS_E(\mathcal{R}, x, y)]=Q(d, 0)+Q(d, 1)$, the lemma will follow.

Let $0\le t \leq d$. For $(x, y)$ in the support of $\Prb_t$ and $i \in \{1, \ldots, 2^t\}$ define 
$\mathsf{I}(t, b, i, x, y):=1$ if $x_i=b$ and $A_t$ queries $x_i$ when run on $x$ not later than it queries an index on which $x$ and $y$ differ, and define $\mathsf{I}(t, b, i, x, y):=0$ otherwise. 
We thus have that
\begin{align}
\label{eq:eq1}
Q(t, b, x, y)=\sum_{i=1}^{2^t} \mathsf{I}(t, b, i, x, y).    
\end{align} 
Consider $0 \leq t \leq d-1$, an $(x, y)$ in the support of $\Prb_t$, bits $b, b' \in \bin$ and $i \in \{1, \ldots, 2^t\}$ such that $x_i=b$. We are interested in a lower bound on the quantity
\begin{align}
\label{eq:eq-def}
F(t, b, b', i, x, y):=\frac{\E[\mathsf{I}(t+1, b', (i,0), x', y')+\mathsf{I}(t+1, b', (i,1), x', y')]}{\E[\mathsf{I}(t, b, i, x, y)]},
\end{align}
whenever the denominator is not $0$. We now describe $F$ in words. $t, b, b', i, x$ and $y$ are fixed. $x_i$ is assumed to be $b$. The denominator is the probability that $\A_t$ queries $x_i$ not later than it queries an index where $x$ and $y$ differ. The numerator is the expected number of $b'$-valued variables in $\{u_i^{(0)}, u_i^{(1)}\}$ that is queried by the simulation of $\A_{t+1}$ inside $\A_t$, not later than the simulation of $\A_{t+1}$ queries an index where $x'$ and $y'$ differ (which, as discussed before, is exactly when $\A_t$ queries an index where $x$ and $y$ differ). Both expectations are over the randomness of $\A_t$, which includes the sampling of $\overline{b}=(b_1, \ldots, b_{2^t})$ and the randomness in $\A_{t+1}$ that is simulated inside $\A_t$. Note that $x'$ and $y'$ are random strings, as they depend on $b_1, \ldots, b_{2^t}$.
\paragraph*{Lower bounding $F$}
We wish to show a lower bound on $F(t, b, b', i, x, y)$. Towards this, let us fix a deterministic decision tree $T$ in the support of $\A_{t+1}$. Furthermore, fix the values of all $b_j$ for $j \neq i$. This fixes all the bits of the string $x'$ except $u_i^{(0)}$ and $u_i^{(1)}$. Now, consider the expression for $F$ where the expectations are conditioned on the above fixings, and are only over the randomness of $b_i$ (notice that $b_i$ determines $u_i^{(0)}, u_i^{(1)}$ and whether $\A_t$ queries $x_i$). Under the above fixing, the action of $T$ on the variables $u_i^{(0)}$ and $u_i^{(1)}$ before it queries an index where $x'$ and $y'$ differ is a deterministic decision tree on these two variables. We assume that the tree is not the empty tree (which in particular implies that $T$ does not query an index where $x'$ and $y'$ differ before it queries any of $u_i^{(0)}$ and $u_i^{(1)}$). Assume further that if one of the two variables is queried and found to be $0$, the other one is not queried (as their NAND is already fixed to $1$, and so the value of $g_d$ is insensitive to the value of the other variable). Under these assumptions there are only two structurally different trees on two variables. The two trees $T_0$ and $T_1$ are given below. Two other trees can be obtained by interchanging the roles  of $u_i^{(0)}$ and $u_i^{(1)}$ in $T_0$ and $T_1$. However, from the symmetry of the NAND function and our distributions, considering $T_0$ and $T_1$ suffices.
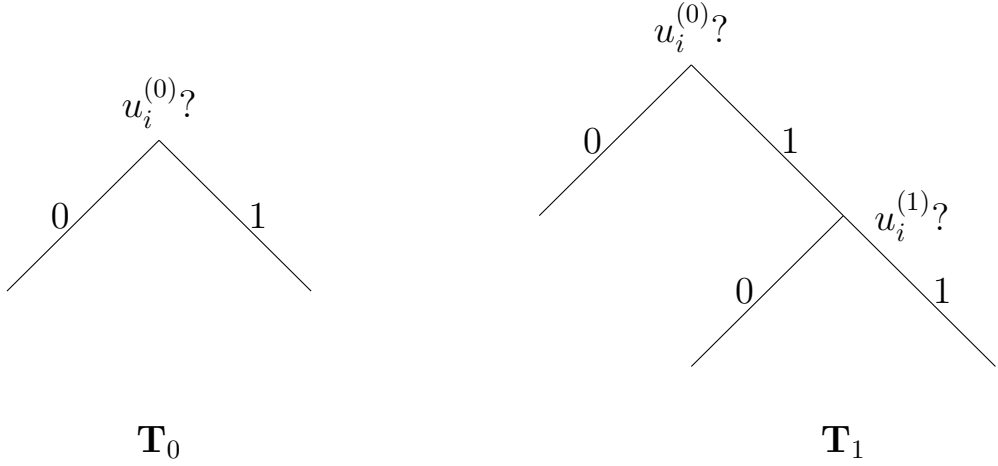
\begin{tikzpicture}
\node at (-0.2, 0){};
\draw (0, 0) -- (2, 2);
\draw (2, 2) -- (4, 0);
\node at (2, 2.5) {\Large $u_i^{(0)}?$};
\node at (0.7, 1){\Large $0$};
\node at (3.3, 1){\Large $1$};
\node at (2, -2){\Large 
 $\mathbf{T}_0$};
\draw (7, 1) -- (9, 3);
\draw (9, 3) -- (11, 1);
\draw (11, 1) -- (9, -1);
\draw (11, 1) -- (13, -1);
\node at (9, 3.5) {\Large $u_i^{(0)}?$};
\node at (11.9, 1){\Large $u_i^{(1)}?$};
\node at (7.7, 2){\Large $0$};
\node at (10.3, 2){\Large $1$};
\node at (12.3, 0){\Large $1$};
\node at (9.7, 0){\Large $0$};
\node at (11, -2){\Large 
 $\mathbf{T}_1$};
\end{tikzpicture}
\\We now show how to bound $F$ for $b=1$ and $b'=0$. Bounds for other combinations can be derived similarly; we list them in Table~\ref{tab:tab1}.

First consider tree $T_0$. Assume that $x_i=b=1$. Consider the denominator of $F$. $A_t$ queries $x_i$ if and only if $T_0$ queries $u_i^{(b_i)}$. $T_0$ queries only $u_i^{(0)}$. Thus, $T_0$ queries $u_i^{(b_i)}$ if and only if $b_i=0$, which happens with probability $1/2$. Thus the denominator is $1/2$.

Now consider the numerator. Number of variables with value $b'=0$ queried by $T$ is $1$ if $u_i^{(0)}=0$ and $0$ otherwise. $u_i^{(0)}=0$ if and only if $b_i=0$, which happens with probability $1/2$. Thus the denominator is $\frac{1}{2}\cdot 1 = 1/2$. Hence, in this case, $F=(1/2)/(1/2)=1$.

Next, consider tree $T_1$. In this case, $x_i$ is guaranteed to be queried, as the tree always queries the variable whose value is $0$. Thus, the denominator is $1$. The numerator is also $1$; exactly one of the two variables is $b'=0$ and $T_1$ stops when it queries a $0$. Thus, in this case too, $F = 1/1 = 1$.

We conclude that when $b=1$ and $b'=0$, a lower bound on $F$ is $\min\{1, 1\}=1$.
\begin{table}
\label{table:table1}
\centering
\begin{tabular}{|c | c | l |} 
 \hline
 $b$ & $b'$ & $F$ \\ [0.5ex] 
 \hline\hline
 $0$ & $0$ & $\ge 0$  \\ 
 $0$ & $1$ & $\ge 2$  \\
 $1$ & $0$ & $\ge 1$  \\
 $1$ & $1$ & $\ge 1/2$  \\
 \hline
\end{tabular}
\caption{Lower bounds on $F$.}
\label{tab:tab1}
 \end{table}
The above analysis holds for a fixed $T$, as long as its restriction to $\{u_i^{(1)}, u_i^{(2)}\}$ until it queries an index where $x'$ and $y'$ differ, is not an empty tree. By averaging, the lower bounds in Table~\ref{tab:tab1} hold for $\A_{t+1}$ and a random $b_1,\ldots, b_{2^t}$ as long as with positive probability the aforementioned restricted tree is not empty.
\paragraph*{A recursive relation for $Q(t, b)$}
Fix $0 \le t\le d-1$, inputs $x, y \in \bin^{2^t}$ such that $(x, y)$ is in the support of $\Prb_{t}$, and bit $b' \in \bin$. Now, consider $Q(t+1, b', x', y')$. Note that $x'$ and $y'$ are random strings, and are determined by $x, y$ (fixed) and $b_1, \ldots, b_{2^t}$ (random). We have that
\begin{align}
\label{eqn:rec1}
Q(t+1, b', x', y')=\sum_{i=1}^{2^t} (\mathsf{I}(t+1, b', (i, 0), x', y')+\mathsf{I}(t+1, b', (i, 1), x', y')).
\end{align}
We split the above sum into two parts depending on $x_i$.
\begin{align}
\label{eqn:rec2}
Q(t+1, b', x', y')=&\sum_{1\leq i\leq 2^t, x_i=0}(\mathsf{I}(t+1, b', (i, 0), x', y')+\mathsf{I}(t+1, b', (i, 1), x', y'))  \nonumber \\
&+\sum_{1\leq i\leq 2^t, x_i=1}(\mathsf{I}(t+1, b', (i, 0), x', y')+\mathsf{I}(t+1, b', (i, 1), x', y')).
\end{align}
Now, we take an expectation on both sides over $b_1, \ldots, b_{2^t}$ and the randomness of $\A_{t+1}$, and apply linearity of expectation.
\begin{align}
\label{eqn:rec3}
\E[Q(t+1, b', x', y')]&=\sum_{1\leq i\leq 2^t, x_i=0}\E[\mathsf{I}(t+1, b', (i, 0), x', y')+\mathsf{I}(t+1, b', (i, 1), x', y') ] \nonumber \\
&+\sum_{1\leq i\leq 2^t, x_i=1}\E[\mathsf{I}(t+1, b', (i, 0), x', y')+\mathsf{I}(t+1, b', (i, 1), x', y')].
\end{align}
Note that if $\E[\mathsf{I}(t, 0, i, x, y)]$ is not $0$, then the summands of the first sum are $F(t, 0, b', i, x, y)\cdot \E[\mathsf{I}(t, 0, i, x, y)]$. A similar statement holds for the second sum. We thus have,
\begin{align}
\label{eqn:rec4}
\E[Q(t+1, b', x', y')]&\ge\sum_{1\leq i\leq 2^t, x_i=0, \mathsf{I}(t, 0, i, x, y)\neq 0}F(t, 0, b', i, x, y)\cdot \E[\mathsf{I}(t, 0, i, x, y)] \nonumber \\
&+\sum_{1\leq i\leq 2^t, x_i=1, \mathsf{I}(t, 1, i, x, y)\neq 0}F(t, 1, b', i, x, y)\cdot \E[\mathsf{I}(t, 1, i, x, y)].
\end{align}
We would now like to consider $b'=0$ and $1$ separately, and plug the bounds of Table~\ref{tab:tab1} into Equation~(\ref{eqn:rec4}). If $\E[\mathsf{I}(t, b, i, x, y)]$ is non-zero, then with positive probability, the restriction of the tree $T$ considered earlier to variables $u_i^{(0)}, u_i^{(1)}$ is not the empty tree; thus the  lower bounds of Table~\ref{tab:tab1} are applicable. We thus have
\begin{align}
&\E[Q(t+1, 0, x', y')] \geq \E[Q(t, 1, x, y)] \label{eq:rel1},\mbox{ and} \\
&\E[Q(t+1, 1, x', y')] \geq 2 \E[Q(t, 0, x, y)] + \frac{1}{2} \E[Q(t, 1, x, y)]. \label{eq:rel2}
\end{align}
Finally, we take expectations over $(x, y) \sim \Prb_t$. As discussed before, this has the effect of inducing the distribution $\Prb_{t+1}$ on $(x', y')$. We thus have
\begin{align}
&Q(t+1, 0) \geq Q(t, 1) \label{eq:finrel1},\mbox{ and} \\
&Q(t+1, 1) \geq 2 Q(t, 0) + \frac{1}{2} Q(t, 1). \label{eq:finrel2}
\end{align}
One can directly check by enumerating all deterministic zero-error trees for $t=0, 1$ that $Q(0, 0), Q(0, 1), Q(1, 0)$ and $Q(1, 1)$ are all $\Omega(1)$. It thus follows from Equations~$(\ref{eq:finrel1})$ and $(\ref{eq:finrel2})$ that $Q(t, b) = \Omega(\alpha^t)$ for $b \in \bin$. In particular, $Q(d, 0), Q(d, 1) = \Omega(\alpha^d)$. This completes the proof of Lemma~\ref{lem:main22}.
\newcommand{\etalchar}[1]{$^{#1}$}

\appendix
\section{Proof of Claim~\ref{clm:rprod>=dprod}}
\label{sec:dprod<=rprod}
In this section we prove Claim~\ref{clm:rprod>=dprod}.
\begin{proof}[Proof of Claim~\ref{clm:rprod>=dprod}]
Let $\mathcal{R}$ be a randomized decision tree that achieves $\R^{\product}_\epsilon(g)$. Fix a product distribution $\mu$. From $\mathcal{R}$, we will construct a deterministic decision tree (with labelled leaves) $T'$ that errs with probability at most $\epsilon$ with respect to $\mu$. This will complete the proof.

To this end, consider any deterministic decision tree $T$ (with unlabelled leaves) in the support of $\mathcal{R}$. We label each leaf $\ell$ of $T$ as follows. Condition $\mu$ on $\ell$ (assume that the conditional probability is defined; otherwise label $\ell$ arbitrarily). If the probability of the event ``$g(x)=1$'' with respect to this conditional distribution is at least $1/2$, we label $\ell$ as $1$. Else, we label $\ell$ as $0$.

In this way we label each leaf of each deterministic decision tree in the support of $\mathcal{R}$. By the guarantee of $\mathcal{R}$, the resulting randomized decision tree (with labelled leaves) computes $g$ on inputs from $\mu$ with error at most $\epsilon$.

Finally, by averaging, it follows that there exists a deterministic tree $T'$ in the support of $\mathcal{R}$ which computes $g$ on a random $x \sim \mu$ with error probability at most $\epsilon$. 
\end{proof}
\end{document}